\RequirePackage{fix-cm}
\documentclass[smallextended]{svjour3}       
\smartqed  
\usepackage{graphicx}

\usepackage[paperwidth=210mm,paperheight=297mm,centering,hmargin=2.3cm,vmargin=2.7cm]{geometry}
\usepackage{marvosym}
\usepackage{amsfonts}
\usepackage{dcolumn}
\usepackage{bm,bbm}
\usepackage{kpfonts}
\usepackage{braket}
\usepackage{xcolor}

\newtheorem{fact}[theorem]{Fact}

\newcommand{\cC}{{\mathcal C}}
\newcommand{\cD}{{\mathcal D}}
\newcommand{\cE}{{\mathcal E}}
\newcommand{\cF}{{\mathcal F}}
\newcommand{\cG}{{\mathcal G}}
\newcommand{\cH}{{\mathcal H}}

\newcommand{\cK}{{\mathcal K}}
\newcommand{\cL}{{\mathcal L}}
\newcommand{\cM}{{\mathcal M}}
\newcommand{\cN}{{\mathcal N}}

\newcommand{\cP}{{\mathcal P}}

\newcommand{\cS}{{\mathcal S}}
\newcommand{\cT}{{\mathcal T}}

\newcommand{\cX}{{\mathcal X}}

\newcommand{\fI}{{\mathfrak I}}

\newcommand{\fM}{{\mathfrak M}}

\newcommand{\fT}{{\mathfrak T}}
\newcommand{\fU}{{\mathfrak U}}

\newcommand{\bbmC}{{\mathbbm C}}
\newcommand{\bbmE}{{\mathbbm E}}
\newcommand{\bbmN}{{\mathbbm N}}

\newcommand{\bbmR}{{\mathbbm R}}

\newcommand{\bbmeins}{{\mathbbm 1}}

\newcommand{\id}{\mathrm{id}}

\newcommand{\tr}{\mathrm{tr}}
\newcommand{\supp}{\mathrm{supp}}

\newcommand{\cl}{\mathrm{cl}}

\newcommand{\conv}{\mathrm{conv}}

\begin{document}

\title{Universal random codes: capacity regions of the compound quantum multiple-access channel	
	   with one classical and one quantum sender \thanks{H. Boche is partly supported by the Deutsche Forschungsgemeinschaft (DFG, German Research
	   	Foundation) under Germany’s Excellence Strategy EXC-2111 390814868 and the Gottfried Wilhelm
	   	Leibniz Prize of the DFG under Grant BO 1734/20-1. G. Janßen is partly supported by the
	   	Bundesministerium für Bildung und Forschung (BMBF, German Federal Ministry of Education and
	   	Research) project QuaDiQua under grant 16KIS0948. S. Saeedinaeeni is partly supported by the BMBF
	   	project Q.Link.X under grant16KIS0858. Parts of the present work were presented at the VDE-TUM
	   	QuaDiQua project meeting at the Walter Schottky Institute, TU Munich (December 2018).}
   }


\titlerunning{Universal random codes \dots}        

\author{Holger Boche         \and
        Gisbert Jan\ss en	 \and
        Sajad Saeedinaeeni
}

\authorrunning{Boche, Jan\ss en, and Saeedinaeeni} 

\institute{H. Boche \at
                Lehrstuhl f\"ur Theoretische Informationstechnik\\
              	Technische Universit\"at M\"unchen \\ 
                80290 M\"unchen \\
                and \\
                Munich Center for Quantum Science and Technology (MCQST), Schellingstr. 4, 80799 München,
                Germany 
                    \and
           		G. Jan\ss en \and S. Saeedinaeeni \at
              	Lehrstuhl f\"ur Theoretische Informationstechnik\\
              	Technische Universit\"at M\"unchen \\ 
              	80290 M\"unchen \\
}


\date{Received: date / Accepted: date}

\maketitle

\begin{abstract}
We consider the compound memoryless quantum multiple-access channel (QMAC) with two sending terminals. In this model, the transmission is governed by the memoryless extensions of a 
completely positive and trace preserving map which can be any element of a prescribed set of possible maps. We study a communication scenario, where one of the senders aims for transmission of classical messages while
the other sender sends quantum information. Combining powerful universal random coding results for classical and quantum information transmission over point-to-point channels,
we establish universal codes for the mentioned two-sender task. Conversely, we prove that the two-dimensional rate region achievable with these codes is optimal. In consequence, we obtain a multi-letter characterization of 
the capacity region of each compound QMAC for the considered transmission task. 
\end{abstract}

\section{Introduction}
 \label{sect:introduction}
 A vast effort in research on quantum communication systems performed over the past few years brings technological applications of quantum communication into sight.
 As a consequence, more realistic communication models have to be considered. Usually, these are are much more involved from the physical as well as the mathematical point of view.
 For example do real-world communication situations usually involve more communication parties than just one sender and one receiver.
 A very basic situation in this field of topics is, when two or more sending parties are connected to a receiver via a multiple-access channel (MAC). A sample use case of this model is, when two senders share the same fiber transmission line to a receiver, while both independently aim to achieve individual transmission goals. Developing coding schemes for such situations is technologically paramount importance, since presuming availability of a "dark fibre" for performing a transmission protocol is rarely feasible. This fact already became apparent as a limiting factor in recent attempts to use commercial fibre lines for quantum key distribution \cite{dynes16} -- commercial fibre lines are usually a valuable resource being shared by many users. \newline    
 Consequently, the rate as well as the  performance each of the sending parties can achieve is 
 in general strongly connected to the signal characteristics of other parties. Finding code constructions which allow to determine the set of rate tuples which are asymptotically achievable in the Shannon-theoretic sense is a highly nontrivial task. \newline 
 Early results which determined the average-error capacity region for classical message transmission over a memoryless classical MAC are due to Ahlswede \cite{ahlswede73} and Liao \cite{liao72}. Among others, these works stimulated a vital 
 research in classical information theory (see \cite{csiszar11}, Chapter 14 for an overview). In case of the maximal error criterion, the capacity region is still unknown while Dueck gave an example of a MAC where average and maximal error 
 capacity regions are different \cite{dueck78}.\newline 
 Regarding the setting, were the senders are connected to the receiver by a memoryless quantum multiple-access channel (QMAC), one of the first notable results is the paper \cite{winter99}. Therein the region of achievable rate pairs 
 was determined for the case that all senders aim to convey classical messages. For other scenario where some of the senders aim to send quantum information, the achievable rates where characterized in \cite{yard08}, while the 
 quantum capacity region (i.e. the set of rate achievable rate pairs if all senders send quantum information) was derived in \cite{horodecki07,yard08}. \newline 
 However, all of the mentioned results were proven under the assumption that the transmission c.p.t.p. map which governs the channel transmission is perfectly known to the senders as well as to the receiver. This assumption is rarely ever fulfilled when a real-world communication system is faced. \newline 
 We therefore impose a channel model based on slightly more realistic assumptions. We consider the \emph{compound memoryless QMAC}, where the communication parties have no precise knowledge of the channel, but instead have only a priori knowledge of 
 a set of channels, in which the generating map is contained. The consequence of their imprecise knowledge is that they have to use ``universal'' codes, which perform well regardless of which channel from the set governs the transmission. In technological applications this "set of confidence" is obtained from prior estimation step, which usually involves sending a pilot signal through the channel and imposing a channel tomography protocol on the output to approximately determine the channel state.    \newline 
 The compound channel model already has been studied in classical Shannon theory since the 1960s. In the domain of quantum Shannon theory, past research activities regarding compound quantum channels where mostly concentrated 
 on point to point quantum channels, leading to determination of the classical capacity \cite{hayashi05,bjelakovic09,datta10,mosonyi15}, and the quantum capacities \cite{bjelakovic08,bjelakovic09_a} of the compound quantum channel. \newline 
 Regarding compound quantum channels having more than two users, the research was concentrated on private classical message transmission over wiretap channels, i.e. channels having one sender but two receiving parties. \newline 
 Suitable codes for this situation were developed in \cite{boche14} generalizing techniques introduced for the case of classical compound wiretap channels in \cite{bjelakovic11} to the quantum setting. \newline 
 The first codes for classical message transmission over compound classical-quantum MACs where provided in \cite{hirche16}. \newline 
 In this article, we aim to extend the scope of multi-user quantum Shannon theory in the direction of models with more than one sender which in addition involve channel uncertainty. \\ 
 We consider a QMAC with two sending parties $A$ and $B$ in 
 the ``hybrid'' situation, where $A$ pursues the target of transmitting classical messages while sender $B$ aims for entanglement transmission or entanglement generation. We determine the capacity region of the compound quantum channel
 model in either case. To prove achievability we invoke rather powerful universal random coding results for classical message transmission and entanglement transmission which are implicitly contained in the literature. With these results 
 at hand we are able to construct suitable random codes for classical/quantum transmission over the compound quantum multiple-access channel. 
 
 \paragraph{Outline} 
 	In Section \ref{sect:defintions_result}, we provide ourselves with precise definitions regarding the channel model and codes used in this work. Therein, we also state Theorem \ref{theorem:main} which is the main result of this work 
 	which is a multi-letter characterization of the capacity region of the compound QMAC with a classical and a quantum sender. 
 	Section \ref{sect:random_coding_preliminaries} is of rather technical nature. We introduce random classical message transmission and entanglement transmission codes which are crucial ingredients for our reasoning. Section
 	\ref{sect:proofs} contains the proofs of Theorem \ref{theorem:main}. In Section \ref{subsect:proofs_inner_bounds} we construct suitable universal hybrid codes for the QMAC. These are  obtained by combining 
 	ideas from \cite{yard08} with the universal random codes from the previous section. By providing the converse part of Theorem \ref{theorem:main} in Section \ref{subsect:proofs_outer_bounds}, we complete our proof. 
 \paragraph{Related work}
 	The capacity regions of a perfectly known QMAC with one classical and one quantum sender (and moreover also the genuine quantum capacity regions of that channel model) where determined by Yard et al. in \cite{yard08}. 
 	The strategy used therein to derive codes being sufficient to prove the coding theorem is as follows. By combining known random coding results for classical message transmission from \cite{holevo98}, \cite{schumacher97}, and 
 	entanglement transmission \cite{devetak05} for single-user quantum channels in a sophisticated way, the authors constructed random codes for classical and quantum coding over the QMAC. Combining random codes to 
 	simultaneously achieve different transmission goals was long standard in classical multi-user Shannon theory, and can, in the quantum case traced back to \cite{winter01}, where the capacity regions of quantum multiple-access channels was determined
 	in case that all senders wish to transmit classical messages. That a strategy in the mentioned manner is successful also in situations where classical and quantum transmission goals are to be accomplished simultaneously 
 	was first demonstrated in \cite{devetak05}. 
 	Therein, hybrid codes are constructed which 
 	allow to transmit classical and quantum information over a memoryless point-to-point quantum channel at the same time. The capacity region for simultaneous transmission of classical and quantum information was also shown to exceed 
 	the obvious region which to be achievable by time-sharing strategies for some quantum channels.\newline 
 	In both cases, suitable random codes already exist in the literature without having been exploited simultaneous transmission yet. In case of universal classical message transmission, independent first results can be found 
 	in \cite{bjelakovic09}, \cite{hayashi05}, and \cite{datta10}. In this work, we exploit the recent and very powerful techniques which where added to the aforementioned results in \cite{mosonyi15}. The random entanglement transmission
 	codes we use in this work where developed in \cite{bjelakovic08}, \cite{bjelakovic09_a}. The reader may note, that the approach pursued in \cite{devetak05} to derive good random quantum codes seems to be not suitable in case of compound quantum channels, 
 	as the discussion in Section VII of \cite{boche14} suggests. The random codes derived in \cite{bjelakovic08} stem from generalization of the codes in \cite{klesse07} which where derived in spirit of the so-called decoupling approach to the quantum
 	capacity (see also \cite{hayden08} for a similar application of the decoupling technique.)

\begin{section}{Notation and conventions} \label{sect:notation}
	All Hilbert spaces which appear in this work are finite dimensional over the field of complex numbers equipped with the standard euclidean scalar product. For a Hilbert space $\cH$, $\cL(\cH)$ denotes the set of linear 
	maps (or matrices), while $\cS(\cH)$ denotes the set of density matrices (states), and $\fU(\cH)$ the set of unitaries on $\cH$. For an alphabet $\cX$ (which we always assume to be of finite cardinality), we denote the simplex of 
	probability distributions on $\cX$ by $\cP(\cX)$. With a second Hilbert space $\cK$, we denote by $\cC^{\downarrow}(\cH,\cK)$ the set of completely positive (c.p.) 
	trace non-increasing maps while $\cC(\cH,\cK)$ is the notation for completely positive and trace preserving (c.p.t.p.) maps. For positive semi-definite matrices $a,b \in \cL(\cH)$, we use the definition
	\begin{align*}
	F(a,b) \ \:= \ \|\sqrt{a}\sqrt{b}\|_1^2
	\end{align*}
	for the (quantum) fidelity. For a c.p.t.p. map $\cN \in \cC(\cH_A, \cH_B)$ and a density matrix $\rho$, we use the entanglement fidelity defined by
	\begin{align*}
	F_e(\rho, \cN) := \braket{\Psi, \id_\cH \otimes \cN(\ket{\Psi}\bra{\Psi}) \Psi},
	\end{align*}
	where $\Psi$ is any purification of $\rho$. The von Neumann entropy of a state $\rho$ is defined by $S(\rho) := -\tr \rho \log \rho$, and we will use in this work several entropic quantities which derive from it. 
	For a bipartite state $\rho \in \cS(\cH_A \otimes \cH_B)$, 
	\begin{align*}
	I_c(A\rangle B, \rho) := S(\rho_B) - S(\rho),
	\end{align*} 
	defines the \emph{coherent information} of $\rho$, while 
	\begin{align*}
	I(A;B, \rho) := S(\rho_A) + S(\rho_B) + S(\rho)
	\end{align*}
	is the \emph{quantum mutual information}. We will employ the usual notation for systems, which have classical and quantum subsystems. E.g. 
	\begin{align*}
	\rho_{XB} := \sum_{x \in \cX} p(x) \ket{x}\bra{x} \otimes \rho_{x} 
	\end{align*}
	represents the preparation of a bipartite system, where one system is classical (with preparation being a probability distribution $p \in \cP(\cX)$) while $\rho_x$ is a density matrix 
	for each outcome $x$ of $X$ (the random variable with probability distribution p.) 
	For a set $A \subset \bbmR_0^+ \times \bbmR_0^+$, we denote the closure of $A$ by $\cl A$. Moreover, we define for each $l \in \bbmN$ the set $\frac{1}{l} A$ by 
	\begin{align*}
	\frac{1}{l}A := \left\{\left(\frac{1}{l}x,\frac{1}{l}y\right): (x,y) \in A\right\}.
	\end{align*}
	For each $\delta > 0$ we moreover set $A_\delta := \left\{x \in \bbmR_0^+ \times \bbmR_0^+: \ \exists y \in A: |x-y| \leq \delta \right\}$.
\end{section}
\begin{section}{Basic definitions and main result} \label{sect:defintions_result}
	\begin{subsection}{Compound QMAC}
		In this section we provide precise definitions of codes and capacity regions considered in this work. 
		Let $\cH_A$, $\cH_B$, $\cH_C$ be Hilbert spaces under control of communication parties labelled by $A$, $B$, and $C$. While $A$, and $B$ act as senders for the channel, $C$ is designated as receiver. 
		Let $\fM \subset \cC(\cH_A \otimes \cH_B, \cH_C)$ be 
		a set of c.p.t.p. maps. Unless  otherwise specified, we do not assume further properties of the set (we especially 
		do not demand $\fM$ to be finite.) \\
		The \emph{compound memoryless quantum multiple access channel (QMAC) generated by $\fM$} (the \emph{compound QMAC $\fM$} for short) is given by the family 
		$\{\cM^{\otimes n}: \ \cM \in \fM \}_{n=1}^\infty$
		of transmission maps. The above definition is interpreted as a channel model, where the transmission statistics for 
		$n$ uses of the system is governed by $\cM^{\otimes n}$, where $\cM$ can be any member of $\fM$. We designate $A$ as sender transmitting classical messages. In this work, we  consider two different coding scenarios 
		which differ in the quantum transmission task $B$ performs, \emph{entanglement generation (EG)} where $B$ and $C$ aim to accomplish entanglement generation, and \emph{entanglement transmission (ET)} where they aim to perform entanglement transmission. 
		For the rest of the section we fix a set $\fM = \{\cM_s^{\otimes n}\}_{s \in S} \subset \cC(\cH_{A} \otimes \cH_B,\cH_C)$ of c.p.t.p. maps. 
		\begin{definition}[EG code]\label{def:scenario-I_code} 
			An \emph{$(n,M_1,M_2)$-EG code} for the compound QMAC $\fM$ is a family $\cC = (V(m), \Psi, \cD_m)_{m=1}^{M_1}$, where with additional
			Hilbert spaces $\cF_B \simeq \cF_C \simeq \bbmC^{M_2}$
			\begin{itemize}
				\item $V: \ [M_1] \ \rightarrow \ \cS(\cH_A^{\otimes n})$ is a classical-quantum channel, i.e. $V(m) \in \cS(\cH_A^{\otimes n})$ for each $m \in [M_1]$.
				\item $\Psi \in \cF_B \otimes \cH_B^{\otimes n}$ is a pure state.
				\item $\cD_m \in \cC^{\downarrow}(\cH_C^{\otimes n}, \bbmC^{M_1} \otimes \cF_C)$ for each $m \in [M_1]$ such that $\sum_{m=1}^{M_1} \cD_m$ is a quantum channel.
			\end{itemize}
		\end{definition}
		In in the situation, where $B$ and the receiver perform entanglement transmission over the QMAC, we define 
		\begin{definition}[ET code]\label{def:ET_code} 
			An \emph{$(n,M_1,M_2)$-ET code} for the compound QMAC $\fM$ is a family $\cC = (V(m), \cE, \cD_m)_{m=1}^{M_1}$, where with additional
			Hilbert spaces $\cF_B \simeq \cF_C \simeq \bbmC^{M_2}$
			\begin{itemize}
				\item $V: \ [M_1] \ \rightarrow \ \cS(\cH_A^{\otimes n})$ is a classical-quantum channel.
				\item $\cE \in \cC(\cF_B,\cH_B^{\otimes n})$
				\item $\cD_m \in \cC^{\downarrow}(\cH_C^{\otimes n}, \bbmC^{M_1} \otimes \cF_C)$ for each $m \in [M_1]$, such that $\sum_{m=1}^{M_1} \cD_m$ is trace preserving.
			\end{itemize}
		\end{definition}
		We next define the performance functions of the codes introduced above
		\begin{align*}
		P^{EG}(\cC, \cM^{\otimes n}, m) \ &:= \ F\left(\ket{m}\bra{m} \otimes \Phi, \id_{\cF_B} \otimes \cD \circ \cM^{\otimes n}(V(m) \otimes \Psi)\right), \text{ and} \\
		P^{ET}(\cC, \cM^{\otimes n}, m) \ &:= \ F\left(\ket{m}\bra{m} \otimes \Phi, \id_{\cF_B} \otimes \cD \circ \cM^{\otimes n}(V(m) \otimes (\cE \otimes \id_{\cF_B})(\Phi))\right).
		\end{align*}
		where $\ket{\Phi} := \sqrt{M_2}^{-1} \sum_{x=1}^{M_2} \ket{x} \otimes \ket{x} \in \cF_A \otimes \cF_A$. We set for $X \in \{ET,EG\}$
		\begin{align*}
		P^{X}(\cC,\cM^{\otimes n}) \ := \ \frac{1}{M_1} \sum_{m=1}^{M_1} \ P^{X}(\cC, \cM^{\otimes n}, m).
		\end{align*}
		\begin{definition}[Achievable rates] \label{def:scenario-ET_achiev_rates}
			Let $X \in \{EG,ET\}$. A pair $(R_1,R_2)$ of non-negative numbers is called an \emph{achievable Scenario-X rate} for the compound QMAC $\fM$, if for each $\epsilon, \delta > 0$
			exists a number $n_0 = n_0(\epsilon, \delta)$, such that for each $n > n_0$ there is an $(n,M_1,M_2)$-Scenario-X code $\cC$ for $\fM$ such that the conditions
			\begin{enumerate}
				\item $\frac{1}{n} \log M_i \geq R_i - \delta$ for $i \in \{1,2\}$, and 
				\item $\inf_{s \in S} \ P^{X}(\cC, \cM_s^{\otimes n}) \geq 1 - \epsilon$ 
			\end{enumerate}
			are simultaneously fulfilled. We define the \emph{Scenario-X capacity region of the compound QMAC $\fM$} by
			\begin{align}
			CQ^{X}(\fM) \ := \ \{(R_1,R_2) \in \bbmR_0^+ \times \bbmR_0^+: \ (R_1,R_2) \ \text{achievable Scenario-X rate for} \ \fM  \}. \label{def:cap_resions}
			\end{align}
		\end{definition}
		The following operational facts, follow directly from the above definitions.
		\begin{fact}[Time Sharing] \label{fact:structure_cap_region}
			For each $\fM \subset \cC(\cH_A \otimes \cH_B, \cH_C)$, $CQ^{EG}(\fM)$ and $CQ^{ET}(\fM)$ are compact convex subsets of $\bbmR^2$.
		\end{fact}
		\begin{fact} \label{prop:ent_transm_ent_gen}
			Let $\fM \subset \cC(\cH_A \otimes \cH_B, \cH_C)$. It holds $CQ^{ET}(\fM) \ \subset \ CQ^{EG}(\fM)$.
		\end{fact}
		\begin{proof}
			Let for an arbitrary but fixed blocklength $n \in \bbmN$, $\cC := (V(m), \cE, \cD_m)_{m=1}^{M_1}$ be an $(n,M_1,M_2)$ ET code. Let for fixed $\cM \in \cC(\cH_A \otimes \cH_B, \cH_C)$, $m \in [M_1]$, a 
			spectral decomposition of $\id_{\cF} \otimes \cE(\Phi)$ given by
			$\sum_{i=1}^N \lambda_i \Psi_i$. It holds for each $m \in [M_1]$
			\begin{align}
			& F\left(\ket{m}\bra{m} \otimes \Phi, \id_{\cF} \otimes \cD \circ \cM^{\otimes n} \circ \id_{\cH_{A}}^{\otimes n} \otimes \cE(V(m) \otimes \Phi)\right)  
			& = \sum_{i=1}^N \lambda_i \ F\left(\ket{m}\bra{m} \otimes \Phi, \id_\cF \otimes \cD_m \circ \cM^{\otimes n}(V(m) \otimes \Psi_i)\right). \label{prop:ent_transm_ent_gen_1}
			\end{align}
			If now $j$ is any index such that $F(\ket{m}\bra{m} \otimes \Phi, \id_\cF \otimes \cD_m \circ \cM^{\otimes n}(V(m) \otimes \Psi_i))$ is maximal, the $(n,M_1,M_2)$ EG code $\tilde{\cC} := (V(m), \Psi_j, \cD_m)_{m=1}^{M_1}$ 
			suffices 
			$P^{EG}(\tilde{\cC}, \cM^{\otimes n}) \ \geq \ P^{ET}(\cC, \cM^{\otimes n})$
			by the inequality in (\ref{prop:ent_transm_ent_gen_1}). \qed
		\end{proof}
		In order to concisely state the main result, we introduce some more notation. Fix Hilbert spaces $\cK_{A}, \ \cK_B, \ \cK_C$, and an alphabet $\cX$. For given probability distribution $p$, c.p.t.p. map $\cT \in \cC(\cK_A \otimes \cK_B, \cK_C)$, 
		cq channel $V: \cX \rightarrow \cS(\cK_A)$, pure state $\Psi \in \cS(\cK_B^{\otimes 2})$, we define an effective cqq state
		\begin{align}\label{def:effective_cqq_state}
		\omega_{XBC} := \omega(\cT, p, V, \Psi) := \sum_{x \in \cX} p(x) \ket{x}\bra{x} \otimes \id_{\cK_B} \otimes \cT(V(x) \otimes \Psi), 
		\end{align}  
		and a region 
		\begin{align} \label{def:one_shot_region}
		\hat{C}^{(1)}(\cT, p, V, \Psi) := \{(R_1, R_2) \in \bbmR_0^+ \times \bbmR_0^+: R_1 \leq I(X;C,\omega) \ \wedge \ R_2 \leq I(B\rangle CX, \omega) \},
		\end{align}
		The following theorem is the first main result this work. It gives a formula for calculation of the capacity regions of a compound quantum multiple-access channel in a situation where one sender send classical messages and the other is performing either entanglement generation or entanglement transmission with the receiver.
		\begin{theorem}\label{theorem:main}
			Let $\fM \subset \cC(\cH_A \otimes \cH_B, \cH_C)$ be a set of quantum channels. The capacity regions $CQ^{EG}(\fM)$ and $CQ^{ET}(\fM)$ of $\fM$ as defined in Eq. (\ref{def:cap_resions}) are determined by the following chain of equalities.
			\begin{align}
			CQ^{EG}(\fM) \ = \ CQ^{ET}(\fM) \ = \ \cl \left(\bigcup_{l=1}^\infty \bigcup_{p,V,\Psi} \bigcap_{\cM \in \fM} \frac{1}{l}  \hat{C}^{(1)}(\cM^{\otimes l}, p, V, \Psi) \right). \label{theorem:main_capacity_region}
			\end{align}
		\end{theorem}
			\begin{remark}The terms on the rightmost sides of the inclusion chains in the above theorem do not need convexification. The corresponding fact for the capacity region of the perfectly known QMAC was already proven in \cite{yard08}.
			For the reader's convenience, we give an argument for the present situation in Appendix \ref{appendix:capacity_regions_convexity}. 
			\end{remark}
			The proof of the equalities in Eq. (\ref{theorem:main_capacity_region}) is split in several parts. Note, that the inequality $CQ^{ET}(\fM) \ \subset \ CQ^{EG}(\fM)$ is Fact \ref{prop:ent_transm_ent_gen}. That the rightmost 
		term in (\ref{theorem:main_capacity_region}) is a subset of $C^{ET}$ is the statement of Proposition \ref{prop:inner_bound_no_csi}. To complete the proof, we show, that $C^{EG}$ does not contain more points than the rightmost term in 
		Proposition \ref{prop:outer_bound_c_csi}.\newline 
		We conclude this section by giving some remarks regarding the interpretation of Eq. (\ref{theorem:main_capacity_region}). The region of points on the right hand side of \ref{theorem:main_capacity_region} is usually strictly smaller than the intersection of the capacity regions of the individual channels in $\fM$ as arising in the scenario when the channel is perfectly known, which reads
			\begin{align*}
			\bigcap_{\cM \in \fM} CQ^{ET}(\{\fM\}) \ = \ \bigcap_{\cM \in \fM} \cl\left(\bigcup_{l=1}^\infty \bigcup_{p,V,\Psi} \frac{1}{l}  \hat{C}^{(1)}(\cM^{\otimes l}, p, V, \Psi) \right). 
			\end{align*}
			This effect is even known from more basic settings than the present one. Already for single-sender compound channels with the uncertainty set having two members, a good code for one of the channels may fail completely for the other possible channel realization. Designing universal codes is therefore inevitable. Channel uncertainty may moreover have crucial impact on the properties of the transmission capacities even int the aforementioned scenaria involving single-sender compound channels. The transmission capacity of a compound channel with the average transmission error as performance criterion obeys no strong onverse while a perfectly known channel always has a strong converse (see \cite{bjelakovic13} and \cite{ahlswede69} for examples of such behaviour.) \newline 
			Beyond these facts, it even can be shown, that in case of a finite classical compound channel with the average error set as performance criterion the following holds. The channel has a strong converse property if and only if there is a member in the channel's uncertainty set which realizes the minimum of the individual capacities of the possible channel realisations. In this case, a "worst channel" determines the capacity (for an extensive discussion, see Chapters 3-5 in \cite{ahlswede15}). However, the existence of a "worst channel" does in general not save the trouble to find universal codes. An optimal code for the worst channel may nevertheless fail on the other channel. Such a channel and other illumating examples can be found in \cite{ahlswede15}, Chapter 3.    
		
		\end{subsection}
	
\end{section}

\begin{section}{Universal random codes for message and entanglement transmission}\label{sect:random_coding_preliminaries}
	In this section, we collect some universal random coding results for entanglement transmission and classical message transmission over single-sender channels. These are essential ingredients for the construction of codes for proving the main result of the paper, Theorem \ref{theorem:main}. Some of the statements below, are already implicitly 
	contained in the literature. However, the random nature of the codes where not explicitly stressed. Some additional properties of these random codes are revealed below, and may also be useful in other occasions. 
	\begin{subsection}{Classical message transmission} \label{subsect:cq_random}
		For the reader's convenience, we first introduce some terminology regarding classical message transmission over classical-quantum channels. A map $V: \cX \rightarrow \cS(\cH)$ with a (finite) alphabet $\cX$ and a Hilbert space $\cH$ is called a \emph{classical-quantum (cq) channel}.
		An \emph{$(n,M)$ classical message transmission code} for $V$ is a family $\cC := (u_m,D_m)_{m=1}^M$, where $u_m \in \cX^n$, and $D_m \in \cL(\cK)$ for each $m\in [M]$, with the additional property, 
		that $\sum_{m=1}^M D_m = \bbmeins_{\cK}$, if $n$ instances of the cq channel $W$ and the code $\cC$ are used 
		for classical message transmission. As error criterion we use the \emph{average transmission error} defined by 
		\begin{align*}
		\overline{e}(\cC, W^{\otimes n}) := \frac{1}{M} \sum_{m=1}^M \ \tr (\bbmeins_\cK^{\otimes n} - D_m) W^{\otimes n}(u_m).
		\end{align*}
		The following proposition states existence of universal random message transmission codes for each given set of classical-quantum channels.
		Its proof can be extracted from \cite{mosonyi15}, where it was proven using the properties of quantum versions of the Renyi entropies together in combination with the Hayashi-Nagaoka random coding lemma \cite{hayashi03}. 
		\begin{proposition}[Universal random cq codes\cite{mosonyi15}, Theorem 4.18] \label{prop:cq_random_coding_no_csi}
			Let $\fI := \{W_t: \cX \rightarrow \cS(\cK_C): t \in T\}$ be a set of classical-quantum channels, and $q \in \cP(\cX)$. For each $\delta > 0$ and large enough $n \in \bbmN$ there exists an $(n,M)$-random 
			message transmission code $\cC(U) = (U_m, D_m(U))_{m=1}^{M}$ which fulfills the following conditions. 
			\begin{enumerate}
				\item $U = (U_1,\dots, U_M)$ is an independent family of random variables, each with distribution $p^{\otimes n}$,
				\item $\tfrac{1}{n} \log M \ \geq \ \inf_{t \in T} I(X;C, \tau_t) - \delta$, where $\tau_t := \sum_{x \in \cX} p(x) \ket{x^X}\bra{x^X} \otimes W_t(x)$, and
				\item $\bbmE \sup_{t \in T} \overline{e}(\cC, W_t) \leq 2^{-nc}$,
			\end{enumerate}
			where $c > 0$ is a constant dependent on $\delta$.
		\end{proposition}
	\end{subsection}
	\begin{subsection}{Entanglement transmission} \label{subsect:q_random}
		In this paragraph we introduce universal coding results for the task of entanglement transmission, which were implicitly proven already in \cite{bjelakovic08},\cite{bjelakovic09_a}. For a given quantum channel 
		$\cN \in \cC(\cK_A, \cK_B)$, an $(n,M)$ entanglement transmission code is a pair $\cC = (\cE,\cD)$, where with a Hilbert space $\cF$ of dimension $M$, $\cE \in \cC(\cF, \cH_A^{\otimes n})$, and $\cD \in \cC(\cH_B^{\otimes n}, \cF)$
		are c.p.t.p. maps. The performance of the code $\cC$ is then measured by the entanglement fidelity $F_e(\pi, \cN)$, where $\pi$ is the maximally mixed state on $\cF$. Their strategy to derive universal entanglement transmission codes
		for compound quantum channels was to generalize the decoupling lemma from \cite{klesse07} to achieve a one-shot bound for the performance in case of a finite set of channels for a fixed code subspace. A subsequent randomization over unitary transformations of that encoding led to random codes 
		with achieving arbitrarily close to coherent information minimized over all possible channel states, given maximally mixed state on the input space. Further approximation using the so-called BSST lemma \cite{bennett02} approximating asymptotically each 
		state by a sequence of maximally mixed states and a net approximation on the set of channels allowed to achieve the capacity of arbitrary compound quantum channels with these random codes. The the authors of \cite{bjelakovic08},
		\cite{bjelakovic09_a} applied a step of derandomization to end up with deterministic entanglement transmission codes sufficient to prove their coding theorem. We in turn, are explicitly interested in the random structure of the code. 
		In the next proposition, we replicate the statement hidden in the proof of  Lemma 9 in \cite{bjelakovic09_a}. Moreover, we notice, that the random codes constructed in that proof have a very convenient property regarding the expected 
		input state to the channel after encoding. It is a tensor product of the maximally mixed state appearing in the coherent information terms lower-bounding the rate. 
		\begin{proposition}[cf. \cite{bjelakovic09_a}, Lemma 9, \cite{boche17}]\label{prop:random_ent_transm_no_csi_finite}
			Let $\fI := \{\cN_t\}_{t \in T} \subset \cC(\cK_A, \cK_B)$ be a set of c.p.t.p. maps, $\cG \subset \cK_A$ a subspace of $\cH$, $\delta > 0$. For each large enough $n$ exists an $(n,M)$ random entanglement transmission code 
			$\cC_u := (\cE_u, \cD_u)$, $u \in A \subset \fU(\cG^{\otimes n})$, $|A| < \infty$ such that 
			\begin{enumerate}
				\item $\frac{1}{n}\log M \  \geq  \ \inf_{t \in T} \ I_c(A\rangle B, \sigma_t) - \delta$, where $\sigma_t := \id_{\cH_A} \otimes \cN_t(\ket{\psi}\bra{\psi})$ with $\psi$ being a purification of $\rho$, 
				\item $\frac{1}{|A|} \sum_{u \in A} \cE_u(\pi_\cF) \ = \ \pi_{\cG}^{\otimes l}$, and 
				\item $\frac{1}{|A|} \sum_{u \in A}  \ \underset{t \in T}{\inf} \ F_e(\pi_\cF, \cD_u \circ \cN^{\otimes n} \circ \cE_u) \ \geq 1 - 2^{-nc}$
			\end{enumerate}
			with a constant $c > 0$. 
		\end{proposition}
		\begin{remark}
			In \cite{bjelakovic08,bjelakovic09_a}, actually a continuous random code distributed according to the Haar measure on the unitary group on the encoding subspace was constructed. For the finite random code in Proposition \ref{prop:random_ent_transm_no_csi_finite} see \cite{boche17}.
		\end{remark}
		We notice, that in earlier work on perfectly known quantum channels (see e.g. \cite{devetak05}, \cite{devetak05b}, \cite{hsieh10}) usually a different type of random code was used. Instead of employing the random entanglement transmission 
		codes from \cite{klesse07} or \cite{hayden08} based on the decoupling approach, the entanglement generation codes of \cite{devetak05} were used. These arise from a clever reformulation of private classical codes for a classical-quantum 
		wiretap channel. In Appendix D in \cite{devetak05}, these codes where further modified to approximately reproduce $\rho^{\otimes n}$ for a given $\rho$ by the random encoding. \newline 
		However, we remark here, that establishing the results on this paper by generalizing the random codes in \cite{devetak05} is not very auspicious. As it was already noted in \cite{boche14}, the method of constructing entanglement generation codes from private classical transmission codes employed in \cite{devetak05}
		seems not to carry over to the case of channel uncertainty. 
		\begin{remark}
			By applying the Proposition \ref{prop:random_ent_transm_no_csi_finite} to the special case $|\fI| = 1$, we obtain an alternative to the random codes from \cite{devetak05} in case of a perfectly known quantum channel.
			Alternatively one could also take the direct route to prove such a result and derive such codes directly from the original works \cite{klesse07}, \cite{hayden08} on the perfectly known channel. 
		\end{remark}
	\end{subsection}
\end{section}
\begin{section}{Proofs}\label{sect:proofs}
	\begin{subsection}{Inner bounds to the capacity regions} \label{subsect:proofs_inner_bounds}
		In this paragraph, we prove the achievability part of Theorem \ref{theorem:main}, i.e. the following statement. 
		\begin{proposition}[Inner bound for the capacity region for uninformed users]\label{prop:inner_bound_no_csi}
			Let $\fM \subset \cC(\cH_A \otimes \cH_B, \cH_C)$. It holds
			\begin{align}
			CQ^{ET}(\fM) \ \supset \ \cl \left(\bigcup_{l=1}^\infty \bigcup_{p,V,\Psi} \bigcap_{\cM \in \fM} \frac{1}{l}  \hat{C}^{(1)}(\cM^{\otimes l}, p, V, \Psi) \right)
			\end{align}
		\end{proposition}
		The main technical steps for proving the above assertion is done in the proof of the following proposition.    
		\begin{proposition} \label{prop:codes_uninformed_users}
			Let $\fT := \{\cT_s\}_{s\in S} \subset \cC(\cK_A \otimes \cK_B, \cK_C)$, $\Phi \in \cS(\cK_B \otimes \cK_B)$ a pure maximally entangled state, $p \in \cP(\cX)$, $V: \cX \rightarrow \cS(\cK_A)$ a channel having only pure outputs. 
			For each $\delta > 0$ exists a number $n_0$ such that for each $n > n_0$ we find an $(n,M_1,M_2)$ Scenario-ET code $\cC$ with 
			\begin{enumerate}
				\item $\tfrac{1}{n} \log M_1 \ \geq \ \inf_{s \in S} \ I(X;C, \omega_s) - \delta$,
				\item $\tfrac{1}{n} \log M_2 \ \geq \ \inf_{s \in S} \ I_c(B\rangle CX, \omega_s) - \delta$, and
				\item $\inf_{s \in S} P^{ET}(\cC, \cT_s) \ \geq \ 1 - 2 ^{-nc}$
			\end{enumerate}
			where $c$ is a strictly positive constant. The entropic quantities on the right hand sides of the first two inequalities are evaluated on the states 
			\begin{align*}
			\omega_s := \omega(\cT_s,p,V,\Psi)  = \sum_{x \in \cX} p(x) \ket{x} \bra{x} \otimes \id_{\cK_B} \otimes \cT_s(V(x) \otimes \Phi)  &&(s \in S)
			\end{align*}
		\end{proposition}
		Before we give a proof of the above proposition, we state a net approximation result which is used therein. We use the diamond norm $\|\cdot\|_\diamond$ defined on the set of maps from $\cL(\cH)$ to $\cL(\cK)$ by
		\begin{align}
		\|\cN\|_\diamond := \underset{a \in \cL(\cH \otimes \cH)}{\max}\|\id_{\cH} \otimes \cN(a)\|_1 &&(\cN \in \cL(\cH)).
		\end{align}
		We will use
		\begin{lemma}[\cite{bjelakovic08}, Lemma 5.2] \label{net_approximation}
			Let $\fI \subset \cC(\cH, \cK)$. For each $\theta > 0$ there is a set $\fI_\theta \subset \fI$ such that the following conditions are fulfilled
			\begin{enumerate}
				\item $|\fI_\theta| \ \leq (6/\theta)^{2 (\dim \cK \cdot \dim\cH)^2}$, and 
				\item for each $\cN \in \fI$ exists $\cN' \in \fI_\theta$ such that $\|\cN - \cN'\|_\diamond \leq \theta$.
			\end{enumerate}
		\end{lemma}
		\begin{proof}
			Fix $\delta > 0$, and set 
			$R_1 
			:= \underset{s \in S}{\inf} \ I(X;C, \omega_s), \hspace{.3cm}$, and 
			$R_2 
			:= \underset{s \in S}{\inf} \ I(B\rangle XC, \omega_s).$
			We assume that $R_1 - \delta$ and $R_2 - \delta$ are both non-negative, otherwise the results follow either by trivial coding or reduction to a case of channel coding with a single sender. We define the state
			$
			\pi_B := \tr_{\cK_B} \Phi,
			$and a classical-quantum channel $T_{A,s}: \cX \rightarrow \cS(\cH_C)$ with outputs 
			$T_{A,s}(x) := \cT_s(V(x) \otimes \pi_B)$ 
			for each $s \in S$.
			If we fix the blocklength $n$ to be sufficiently large, we find by virtue of Proposition \ref{prop:cq_random_coding_no_csi} a random $(n,M_1)$ message transmission code $\cC_A(U) := (U_m, \hat{D}_m(U))_{m=1}^{M_1}$ for 
			the compound classical-quantum channel $\{T_{A,s}\}_{s \in S}$, where $U = (U_1,\dots,U_{M_1})$ is an i.i.d. random family with generic distribution $p^{\otimes n}$, rate
			$ 
			\frac{1}{n} \log M_1 \ \geq \ I(X;C, \omega_s) - \delta \ = \ R_1 - \delta,
			$ 
			and expected average message transmission error
			\begin{align}
			\bbmE \ \overline{e}(\cC_A(U), T_{A,s}^{\otimes n}) \leq 2^{-nc_1} \label{prop:codes_uninformed_users_initial_exp_av_err}
			\end{align}
			where $c_1$ is a strictly positive constant.
			Moreover we define for each $s$ a c.p.t.p. map $\cT_{B,s} \in \cC(\cK_B, \cK_B \otimes \bbmC^{|\cX|})$ by
			\begin{align*}
			\cT_{B,s}(\tau) := \sum_{x \in \cX} p(x) \cT_s(V(x) \otimes \tau) \otimes \ket{x} \bra{x} &&(\tau \in \cL(\cK_B)). 
			\end{align*}
			Under the assumption of large enough blocklength $n$, Proposition \ref{prop:random_ent_transm_no_csi_finite} assures us, that there exists a random $(n, M_2)$ entanglement transmission code 
			$\cC_B(\Lambda) := (\cE_{\Lambda}, \widetilde{\cD}_{\Lambda})$ for the compound quantum channel $\{\cT_{B,s}\}_{s \in S}$ where $\Lambda$ is supported on a 
			finite set $A$, and which has rate
			$\frac{1}{n} \log M_2 \ \geq \ R_2 - \delta$, 
			such, that with a positive constant $c_2 > 0$ the expected entanglement fidelity can be bounded as
			\begin{align}
			\bbmE_\Lambda \ \left[ \ \underset{s \in S}{\inf} \  F_e(\pi_{\cF}, \widetilde{\cD}_\Lambda \circ \cT_{B,s}^{\otimes n} \circ \cE_\Lambda) \right]\ \geq \ 1 - 2^{-nc_2}. \label{prop:codes_uninformed_users_exp_enc_dist}
			\end{align}
			In addition, the expected density matrix resulting from the random encoding procedure is maximally mixed, i.e. 
			$ 
			\bbmE_{\Lambda} \cE_{\Lambda}(\pi_\cF) =  \pi_B^{\otimes n}.
			$ 
			For each pair $(u, \alpha) \in \cX^{nM_1} \times A$ of realizations of $(U, \Lambda)$, we define an $(n,M_1,M_2)$ ET code 
			$ 
			\cC(u,\alpha) := (V(u_m),\cE_\alpha, \cD_{m,\alpha u})_{m=1}^{M_1} 
			$ 
			using the decoding operations
			\begin{align*}
			\cD_{m,\alpha u}(x) \ := \ \widetilde{\cD}_{\alpha}\circ \hat{\cD}_m(u)(x) &&(x \in \cL(\cK_{C}^{\otimes n})).
			\end{align*}
			with $\hat{\cD}_m(u) := \hat{D}_m^{\tfrac{1}{2}}(\cdot)\hat{D}_m^{\tfrac{1}{2}}$.
			Each of the codes defined above already is a Scenario-ET code of suitable rates for classical message and entanglement transmission. To complete the proof of the proposition, we will lower-bound the  the expected ET fidelity 
			of the random code $\cC(U,\Lambda)$. Fix, for the moment the channel state $s$. Let $T'_{A,s\alpha}$ be the cq channel defined by the states
			\begin{align}
			T'_{A,s\alpha}(x^n) := \cT_{s}^{\otimes n}(V^{\otimes n}(x^n) \otimes  \cE_{\alpha}(\pi_\cF)) &&(\alpha \in A, x^n \in \cX^n). \label{prop:codes_uninformed_users_eff_cq_def_2}
			\end{align}
			Averaging over the random choice of $\alpha$ the transmission statistics of $T_{A,s}^{\otimes n}$ is reproduced. Indeed, for each $x^n \in \cX^n$ 
			\begin{align}
			\bbmE_{\Lambda} T'_{A,s\Lambda}(x^n) \ = \ \cT_{s}^{\otimes n}(V^{\otimes n}(x^n) \otimes  \bbmE_{\Lambda} \cE_{\alpha}(\pi_\cF)) \ = \  T_{A,s}^{\otimes n}(x^n). \label{prop:codes_uninformed_users_exp_channel_dist}
			\end{align}
			Since the average transmission error is an affine function of the cq channel, we have for each $u \in \cX^{nM_1}$
			\begin{align}
			\bbmE_U\bbmE_{\Lambda} \ \overline{e}(\cC(u), T'_{A,s\Lambda})  \
			& =  \ \bbmE_U \overline{e}(\cC(u), T_{A,s}^{\otimes n}) \leq 2^{-nc_1}. \label{prop:codes_uninformed_users_class_err}
			\end{align}
			Define the cq channel $T_{s,\alpha}$ defined by 
			\begin{align}
			T'_{s\alpha}(x^n) := \id_{\cK_B}^{\otimes n} \otimes \cT_s^{\otimes n}(V^{\otimes n}(x^n) \otimes (\id_{\cK_B}^{\otimes n} \otimes \cE_\alpha(\Phi))) &&(x^n \in \cX^n, \alpha \in A)
			\end{align}
			(note that the reduction of $T'_{s\alpha}(x^n)$ is, in fact, $T'_{A,s\alpha}$.) For each classical message $m \in [M_1]$, it holds
			\begin{align}
			\bbmE_{\Lambda} \bbmE_U \ \tr\left(\id_{\cK_B}^{\otimes n}\otimes \hat{\cD}(U)(T'_{s\Lambda}(U_m))\right) \ 
			& = \ \bbmE_{\Lambda}\bbmE_U \ \tr \hat{D}_m(U) T'_{A,s\Lambda}(U_m) \geq 1 -  2^{-nc_1}.
			\end{align}
			The above inequality stems from the bound in (\ref{prop:codes_uninformed_users_class_err}) together with the observation, that by symmetry of the random selection procedure for the codewords, 
			the expectation of the one-word message transmission error does not depend on the individual message $m$. If we define 
			\begin{align}
			\gamma_s(u,\alpha) \ := \ 1 - \tr (\bbmeins_{\cK_B}^{\otimes n} \otimes \hat{D}_m(u))T'_{A,s \alpha}(u_m)  \label{prop:codes_uninformed_users_class_exp_err}
			\end{align}
			we have by (\ref{prop:codes_uninformed_users_class_exp_err}) 
			\begin{align}
			\bbmE_\Lambda \bbmE_U \ \gamma_s \ \leq 2^{-nc_1}. \label{prop:codes_uninformed_users_class_err_2}
			\end{align}
			Define for each realization $(u, \alpha)$ of $(U,\Lambda)$
			\begin{align}
			\Gamma^{(s)}_{mm'}(u, \alpha) \ 
			&:= \ \id_{\cK_B}^{\otimes n} \otimes \hat{\cD}_m(U)( T'_{A,s\alpha}(u_m))   &&(m,m' \in [M_1]), \\
			\hat{\Gamma}_m^{(s)}(u,\alpha) \
			&:= \ \sum_{m'=1}^{M_1} \Gamma_{mm'}^{(s)}(u, \alpha) \otimes \ket{u_{m'}}\bra{u_{m'}}, \ \text{and} \\
			\Gamma'^{(s)}(x^n, \alpha) \ &:= \ \id_{\cK_B}^{\otimes n} \otimes \cT_s^{\otimes n}(V^{\otimes n}(x^n) \otimes (\cE_{\alpha}(\Phi)) \otimes \ket{x^n} \bra{x^n} &&(x^n \in \cX^n).
			\end{align}
			Note, that if $\tilde{U}$ is an $\cX^n$-valued random variable with distribution $p^{\otimes n}$, that 
			\begin{align}
			\bbmE_{\tilde{U}}\  \Gamma'^{(s)}(\tilde{U}, \alpha) \ = \ (\id_{\cK_B}^{\otimes n} \otimes \cT_{B,s}^{\otimes n} \circ \cE_{\alpha})(\Phi) \label{prop:codes_uninformed_users_trbs_est}
			\end{align}
			holds by definition of $\cT_{B,s}$. By the gentle measurement lemma (see Lemma \ref{lemma:gentle_measurement} in  Appendix \ref{sect:auxiliary_results}), we have
			\begin{align}
			\|\id_{\cK_B}^{\otimes n} \otimes  (\id_{\cH_C}^{\otimes n} - \hat{\cD}_m(u))T'_{A,s\alpha}(u_m)\|_1 \ \leq \ 3 \sqrt{\gamma_s(u, \alpha)}.
			\label{prop:codes_uninformed_users_normone_deviation}
			\end{align}
			Taking expectations on both sides of the inequality in (\ref{prop:codes_uninformed_users_normone_deviation}), we arrive at
			\begin{align*}
			\bbmE_\Lambda \bbmE_U \|\id_{\cK_B}^{\otimes n} \otimes  (\id_{\cH_C}^{\otimes n} - \hat{\cD}_m(U))T'_{A,s\alpha}(U_m)\|_1\ 
			\leq \ 3 \bbmE_\Lambda \bbmE_U \sqrt{\gamma_s(u, \alpha)}  
			\leq 3 \sqrt{\bbmE_\Lambda \bbmE_U \gamma_s(u, \alpha)} 
			\leq 3\sqrt{2^{-nc_1}}.
			\end{align*}
			where the second inequality above is by Jensen's inequality together with concavity of the square-root function. The last inequality is by the estimate in (\ref{prop:codes_uninformed_users_class_err_2}). As a consequence of 
			these bounds, we have 
			\begin{align}
			\bbmE_U \bbmE_\Lambda  \ \|\hat{\Gamma}^{(s)}_m(U, \Lambda) - \Gamma'^{(s)}(U_m, \Lambda)\|_1 \ \leq \ 3 \cdot \sqrt{2^{-nc_1}} + 2^{-nc_1} \ \leq \ 4 \cdot \sqrt{2^{-nc_1}}. \label{prop:codes_uninformed_users_tracedistance}
			\end{align}
			for each message $m \in [M_1]$. Moreover, we can bound 
			\begin{align}
			\bbmE_{\Lambda}\bbmE_U \ F\left(\Phi, \id_{\cK_B}^{\otimes n} \otimes \widetilde{\cD}_{\Lambda}(\hat{\Gamma}^{(s)})\right) 
			& \ \geq \ \bbmE_{\Lambda}\bbmE_U \ \left( F\left(\Phi, \id_{\cK_B}^{\otimes n} \otimes \widetilde{\cD}_{\Lambda}(\Gamma'^{(s)})\right)  
			-  \ \|\id_{\cK_B}^{\otimes n} \otimes \widetilde{\cD}(\Gamma^{(s)}(U,\Lambda) - \Gamma'^{(s)}(U_m, \Lambda)) \|_1\right) \nonumber \\
			& \geq \bbmE_\Lambda \bbmE_U  \ \left( F_e\left(\Phi, \tilde{\cD}_{\Lambda}\circ \cT_{B,s}^{\otimes n} \circ \cE_{\Lambda}\right) 
			- \ \|\Gamma^{(s)}(U,\Lambda) - \Gamma'^{(s)}(U_m, \Lambda) \|_1\right)    \nonumber \\
			& \geq \bbmE_{\Lambda} \ F_e(\Phi, \tilde{\cD}_{\Lambda}\circ \cT_{B,s}^{\otimes n} \circ \cE_{\Lambda}) - 4\cdot \sqrt{2^{-nc_1}} \nonumber  \\
			& \geq 1 - 2^{-nc_2} - 4\cdot \sqrt{2^{-nc_1}}. \label{prop:codes_uninformed_users_quant_error}
			\end{align}
			The first line above is by application of Lemma \ref{lemma:auxiliary_results_yard_fidelity_1} which can be found in Appendix \ref{sect:auxiliary_results}. The second is by using the equality in (\ref{prop:codes_uninformed_users_trbs_est}) together with monotonicity of the trace 
			distance under taking partial traces. The third 
			is by (\ref{prop:codes_uninformed_users_tracedistance}). The last estimate comes from (\ref{prop:codes_uninformed_users_exp_enc_dist}).
			Putting all the estimates together, we can bound the expected ET fidelity. We have for each $m \in [M_1]$
			\begin{align}
			\bbmE_U \bbmE_\Lambda \ P^{ET}\left( \cC(U,\Lambda), \cT_s^{\otimes n}, m\right) 
			& \ = \ \bbmE F\left(\ket{m}\bra{m} \otimes \Phi, \id_{\cK_B}^{\otimes n} \otimes \cD_{U, \Lambda}\circ \cT_s^{\otimes n} \circ (\id_{\cK_A}^{\otimes n} \otimes \cE_\Lambda)(V^{\otimes n}(U_m) \otimes \Phi)\right) \nonumber \\
			& \ \geq 1 - \bbmE \|\ket{m} \bra{m} - \tr_{\cF} \circ \cD_{U\Lambda} \circ \cT_s^{\otimes n}(V^{\otimes n}(U_m) \otimes \cE_\Lambda(\pi))\|_1 \nonumber \\
			&  - 3 \cdot \left(1 - \bbmE sF(\Phi, \id_{\cK_B}^{\otimes n} \otimes \tr_{\bbmC^{M_1}} \circ \cD_{U,\Lambda} \circ \cT_s^{\otimes n} \circ (\id_{\cK_A}^{\otimes n} \otimes \cE_{\Lambda})(V^{\otimes n}(U_m) \otimes \Phi)  \right) 
			\nonumber \\
			& \geq 1 - 2 \bbmE \overline{e}(\cC, \cT'_{A,s \Lambda}) - 3 \left(1 - F(\Phi, \bbmE_{\Lambda}(\id_{\cK_B}^{\otimes n} \otimes \tilde{\cD}_\Lambda(\bbmE_U \hat{\Gamma}^{(s)}(U,\Lambda))))\right) \nonumber \\
			& \geq 1 - 2 ^{-nc_1} + 2 \cdot  2^{-nc_2} + 4 \sqrt{2^{-nc_1} + 2^{-nc_2}}. 
			\end{align}
			The first inequality is by Lemma \ref{lemma:auxiliary_results_yard_fidelity_2} to be found in Appendix \ref{sect:auxiliary_results}. The third inequality is by inserting the bounds from (\ref{prop:codes_uninformed_users_class_err}) and (\ref{prop:codes_uninformed_users_quant_error}).
			Consequently, we have for each $s \in S$
			\begin{align}
			\bbmE P^{ET}(\cC(U,\Lambda), \cT_s^{\otimes n}) \geq  1 - 2^{-n \frac{\tilde{c}}{2}} \label{prop:codes_uninformed_users_pre_last}
			\end{align}
			with  $\tilde{c} := \min\{c_1,c_2\}$ provided, that $n$ is large enough. The inequality in (\ref{prop:codes_uninformed_users_pre_last}, in fact provides an individual lower bound on the expected ET fidelity for each channel state 
			$s$. Since we aim for a lower bound on the expected worst-case fidelity over the set $S$, we include another step of approximation. We derive from the individual bounds on the expected ET fidelity of the 
			random code for each $s$ a universal bound, i.e. a bound on the expected worst-case ET fidelity of the code. We assume, for each $n \in \bbmN$,  $\tilde{S}_n$ to be a subset of $S$, 
			such that $\{\cT_s\}_{s \in \tilde{S}}$ is a $\eta_n$-net for the original set of channels
			generating the transmission with $\delta_n := 2^{-n\hat{c}}$, i.e. to each $s \in S$ exist an $\tilde{s} \in \tilde{S}_n$, such that $\|\cT_s - \cT_{\tilde{s}}\|_{\diamond} \ \leq \ \delta_n$. 
			We assume moreover, that the cardinality of each of these sets is bounded by $|\tilde{S}_n| \ \leq \ 2^{n \tfrac{\tilde{c}}{4}}$. Note, that such sets indeed exist by Lemma \ref{net_approximation}. We have 
			\begin{align}
			\bbmE_{U}\bbmE_{\Lambda} \left( \frac{1}{|\tilde{S}|}\sum_{s \in \tilde{S}} \  P^{ET}(\cC(U,\Lambda), \cT_s^{\otimes n}) \right)\
			&= \ \frac{1}{|\tilde{S}|}\sum_{s \in \tilde{S}} \bbmE_U \bbmE_\Lambda \  P^{ET}(\cC(U,\Lambda), \cT_s^{\otimes n})  
			\geq \ 1-2^{-n \tfrac{\tilde{c}}{2}} \label{prop:codes_uninformed_users_pre_last_last}
			\end{align}
			by (\ref{prop:codes_uninformed_users_pre_last}), which implies,  
			\begin{align}
			\bbmE_{U} \bbmE_{\Lambda} \  \underset{s \in \tilde{S}_n}{\min}\ P^{ET}(\cC(U,\Lambda), \cT_s^{\otimes n}) \ \geq 1 \ - |\tilde{S}_n| \ \cdot 2^{-n \tfrac{\tilde{c}}{2}} \geq 1 - 2^{-n \tfrac{\tilde{c}}{4}}.
			\end{align}
			for each sufficiently large $n$. The rightmost inequality above follows from the cardinality bound in 
			(\ref{prop:codes_uninformed_users_pre_last}). By continuity of $P^{ET}$, we have
			\begin{align}
			\bbmE_{U} \bbmE_\Lambda \  \underset{s \in S}{\inf} \ P^{ET}(\cC(U,\Lambda), \cT_s^{\otimes n}) \ 
			\geq \ \bbmE_{U} \bbmE_\Lambda \ \underset{s \in S}{\inf} \ P^{ET}(\cC(U,\Lambda), \cT_s^{\otimes n}) - n \delta_n  
			\ \geq \ 1- 2^{-n c}
			\end{align}
			with a strictly positive constant $c$. We are done.
		\end{proof}
		\newpage 
		Next, we prove a generalization of Proposition \ref{prop:codes_uninformed_users}, where we drop the condition, of $\Phi$ being a maximally entangled state. 
		\begin{proposition} \label{prop:codes_uninformed_users_2}
			Let $\fT := \{\cT_s\}_{s\in S} \subset \cC(\cK_A \otimes \cK_B, \cK_C)$, $\Psi \in \cS(\cK_B \otimes \cK_B)$ a pure state, $p \in \cP(\cX)$, $V: \cX \rightarrow \cS(\cK_A)$ a channel with pure outputs. For each $\delta > 0$ 
			exists a number $n_0$ such that for each $n > n_0$ we find an $(n,M_1,M_2)$ ET code $\cC$ with 
			\begin{enumerate}
				\item $\tfrac{1}{n} \log M_1 \ \geq \ \inf_{s \in S} \ I(X;C, \omega_s) - \delta$,
				\item $\tfrac{1}{n} \log M_1 \ \geq \ \inf_{s \in S} \ I(B\rangle CX, \omega_s) - \delta$, and
				\item $\inf_{s \in S} P^{ET}(\cC, \cT_s) \ \geq \ 1 - 2 ^{-nc}$
			\end{enumerate}
			where $c$ is a strictly positive constant. 
		\end{proposition}
		to prove the above statement, we will invoke Proposition \ref{prop:codes_uninformed_users}, together with an elementary approximation argument. Note, that for each $\rho \in \cS(\cH)$, $l \in \bbmN$, $\rho^{\otimes l}$, 
		can be written in the form 
		\begin{align}
		\rho^{\otimes l}  = \sum_{i=1}^D q(i) \pi_i, \label{typical_decomposition_eq}
		\end{align}
		with maximally mixed states $\pi_1, \dots, \pi_D$ supported on pairwise mutually orthogonal subspaces, $q \in \cP([D])$, and $D \leq (l+1)^{\dim \cH}$. This can be seen by using the spectral decomposition together with 
		the fact, that the spectrum of $a^{\otimes l}$ has its cardinality upper-bounded by $(l+1)^{\dim \cH}$ for each $a \in \cL(\cH)$.
		The following lemma formalizes a very basic fact about approximation by empirical distributions.
		\begin{lemma} \label{lemma:approx_emp}
			Let $q \in \cP(\cX)$ a probability distribution. There exists a constant, for each $t \in \bbmN$, $t > 2/\min_{x: p(x) > 0} p(x)$  exist integers $N_x \in \bbmN$, $x \in \cX$, such that $N_x$ is zero if $p(x)$ vanishes, and 
			\begin{itemize}
				\item $\forall x \in \cX: \ |p(x) - N_x/t| \ < |\supp(p)|/t$,  \ \text{and} 
				\item $\forall x \in \supp(p): N_x \geq \ C_p\cdot t$
			\end{itemize}
			with a constant $C_p > 0$.
		\end{lemma}
		\begin{proof}
			Excluding a trivial case, we assume $ p > 1$. Let $x_0$ be any member of $\supp p$. Fix $m \in \bbmN$ large enough, define $N_x := \lfloor m\cdot p(x) \rfloor$ for each $x \neq x_0$, and $N_{x_o} := m - \sum_{x \neq x_0} N_x$.
			That the first claim of the lemma holds follows from these definitions. For the second claim, notice, that for all $x \in \supp p$, it holds $N_x \geq m \cdot p(x) - 1 \geq m \cdot C_p$, where
			$C_p := \min_{x: p(x) > 0} p(x) /2$.
		\end{proof}
		\begin{proof}[Proof of Proposition \ref{prop:codes_uninformed_users_2}]
			Fix $l \in \bbmN$ such that 
			$ 
			\dim \cH \cdot \log(l+1)/l \ \leq \ \tfrac{\delta}{3},
			$ 
			and let $\sum_{i=1}^D q(i) \pi_i$ be a decomposition of $\rho^{\otimes l}$ as in Eq. (\ref{typical_decomposition_eq}). Fix for each $i$ a maximally entangled state $\Phi_i$ which purifies $\pi_i$. 
			and define a state
			\begin{align}
			\omega_{s,i} := \omega_s(\cT_s^{\otimes l}, p^l, V^{\otimes l}, \Phi_i)
			\end{align}
			for each $s, i$. Note, that $\omega_s = \sum_{i=1}^D q(i) \omega_{s,i}$ holds. Using an approximation of $q$ by numbers $N_1, \dots, N_D$ according to Lemma \ref{lemma:approx_emp}, 
			we conclude, that if $t$ is any appropriately large number
			\begin{align}
			I_c(B^l\rangle C^lX^l, \omega_s)  \ 
			& = \frac{1}{l} I_c(B^l\rangle C^lX^l, \omega_s^{\otimes l}) \\
			&\leq  \ \sum_{i=1}^D q(i) \frac{1}{l}I(B^l\rangle C^lX^l, \omega_{s,i}) + \dim \cH \frac{\log (l+1)}{l}  \\
			& \leq \ \sum_{i=1}^D \frac{N_{i}}{l\cdot t} I(B^l\rangle C^lX^l, \omega_{s,i}) + \frac{\delta}{2}  \label{coherent_approx}
			\end{align}
			does hold. The first of the above inequalities is by almost-convexity of the von Neumann entropy.  Let for each $i \in [D]$, $\cC_i := (V_i(m), \cE_i , \cD^{(i)}_m)_{m=1}^{M_{i,1}}$ be an $(N_i,M_{i,1},M_{i,2})$-ET code for $\fT$, such that 
			\begin{align}
			\frac{1}{N_i}\log M_{i,1} \geq \underset{s \in S}{\inf} \ \frac{1}{l} I(X^l;C^l, \omega_{s,i}) - \frac{\delta}{4} := R_{1,i}, \  \text{and} \label{code_rate_nu_1}
			\end{align}
			\begin{align}
			\frac{1}{N_i} \log M_{i,2} \ \geq \ \underset{s \in S}{\inf} \ \frac{1}{l}I_c(B^l\rangle C^lX^l, \omega_{s,i}) - \frac{\delta}{4} := R_{2,i} \label{code_rate_nu_2}
			\end{align}
			and 
			\begin{align}
			\underset{s \in S}{\inf} \ P^{ET}(C_i, \cT_s^{lN_i}) \geq 1 - 2^{-N_ic_i} \geq 1 - 2^{-t N c}
			\end{align}
			with $c_i > 0$, $c := \min\{c_1,\dots,c_D\}$. Note, that such codes exist, if we choose $t$ large enough, since the second claim of Lemma \ref{lemma:approx_emp} guarantees long enough blocklengths $N_i$. By concatenation, we build 
			the $(l\cdot t, M_{1},M_{2})$ Scenario II code $\cC := (V(m),\cE, \cD_m)_{m=1}^M$ for $\fT$ with
			\begin{align}
			M_{1} = \prod_{i=1}^D M_{i,1} \ \hspace{1.5cm} \text{and} \hspace{1.5cm} M_{2} = \prod_{i=1}^{D} M_{2,i},  \label{concatenated_code_numbers}
			\end{align}
			where we defined, with any bijection $\iota: [M_1] \rightarrow \prod_{i=1}^D M_{i,1}$,
			\begin{align*}
			V(m) :=   \bigotimes_{i=1}^{M_1} V_i(\iota_i(m)), \hspace{1cm} \cE := \bigotimes_{i=1}^D \cE_i, \ \text{and} \hspace{1.5cm} \cD_m := \bigotimes_{i=1}^{M_1} \cD^{(i)}_{\iota_i}(m).
			\end{align*}
			For the rates of this code, it holds
			\begin{align*}
			\frac{1}{t} \log M_{1} \ 
			&= \ \sum_{i=1}^D \frac{1}{t}\log M_{i,1} \\ 
			&= \ \sum_{i=1}^D \frac{N_i}{t} \frac{1}{N_i}\log M_{i,1} \\  
			&\geq  \sum_{i=1}^D \frac{N_i}{t} \frac{1}{l}I(X^l;C^l, \omega_{s,i}) - \frac{\delta}{4l} \\ 
			&\geq  \sum_{i=1}^D q(i) \frac{1}{l}I(X^l;C^l, \omega_{s,i}) - \frac{\delta}{4l} - \frac{\delta}{2} \\ 
			&\geq  I(X;C, \omega_{s}) - \frac{\delta}{4l} - \frac{\delta}{2}.
			\end{align*}
			The first inequality above is by (\ref{code_rate_nu_1}), the second by choice of $N_1,\dots, N_D$. The last line is by convexity of the quantum mutual information. Moreover, we have
			\begin{align*}
			\frac{1}{lt} \log M_{2} \ 
			&= \ \sum_{i=1}^D \frac{1}{t}\log M_{i,2} \\ 
			&= \ \sum_{i=1}^D \frac{N_i}{t} \frac{1}{N_i}\log M_{i,2} \\  
			&\geq  \sum_{i=1}^D \frac{N_i}{m} \frac{1}{l}I(B \rangle X^lC^l, \omega_{s,i}) - \frac{\delta}{4l} \\ 
			&\geq  I(B\rangle XC, \omega_{s}) - \frac{\delta}{4l} - \frac{\delta}{2},
			\end{align*}
			where the last inequality is by (\ref{coherent_approx}). To evaluate the ET fidelity of $\cC$ regarding $\fT$, we have for each $s \in S$
			\begin{align*}
			P^{ET}(\cC, \cT_s^{l\cdot t}) \ 
			&= \ \frac{1}{M_1} \sum_{m=1}^{M_1} F(\ket{m}\bra{m} \otimes \Phi_m, \id_{\cF_{B}} \otimes \cD_{m} \circ \cT_s^{\otimes lt}(V(m) \otimes \Phi)) \\
			&= \prod_{i=1}^D \frac{1}{M_{1,i}} \sum_{\iota_i(m)=1}^{M_{1,i}} F(\ket{\iota_i(m)}\bra{\iota_i(m)} \otimes \Phi_{i,\iota_i(m)}, \id_{\cF_{B,i}} \otimes \cD^{(i)}_{\iota_i(m)} \circ \cT_s^{\otimes lN_i}(V(\iota_i(m)) \otimes \Phi_{i,\iota_i(m)})) \\
			&\geq (1 - 2^{lt \tilde{c}})^D \\
			&\geq 1- D 2^{-lt \tilde{c}}.
			\end{align*}
			The last inequality is Bernoulli's. We are done.
		\end{proof}
		With the preparation gathered so far, we are able to prove Proposition \ref{prop:inner_bound_no_csi}. 
		\begin{proof}[of Proposition \ref{prop:inner_bound_no_csi}]
			Since $CQ^{ET}(\fM)$ is closed by definition (see Fact \ref{fact:structure_cap_region}), it suffices to show the inclusion
			\begin{align}
			CQ^{ET}(\fM) \ \supset \ \left(\bigcup_{l=1}^\infty \bigcup_{p,V,\Psi} \bigcap_{\cM \in \fM} \frac{1}{l}  \hat{C}^{(1)}(\cM^{\otimes l}, p, V, \Psi) \right). \label{prop:inner_bound_no_csi_1}
			\end{align}
			We fix $l \in \bbmN$, a finite alphabet $\cX$, a cq channel $W: \cX \rightarrow \cS(\cH_A^{\otimes l})$ with pure-state outputs, a pure state $\Psi \in \cS(\cK_B^{\otimes l} \otimes \cK_B^{\otimes l})$, and a probability distribution 
			$p \in \cP(\cX)$. We show
			\begin{align}
			\bigcap_{\cM \in \fM} \frac{1}{l}  \hat{C}^{(1)}(\cM^{\otimes l}, p, V, \Psi) \subset CQ^{ET}(\fM), \label{prop:inner_bound_no_csi_2}
			\end{align}
			which, with subsequent maximization proves the inclusion in (\ref{prop:inner_bound_no_csi_1}). Fix $\delta > 0$, and let $n > l$ be a blocklength written as $n= a\cdot l + b$ with $a,b \in \bbmN$, $0 \leq b < k$. Let $n$ (and consequently $a$) 
			be large enough, to find, using Proposition \ref{prop:codes_uninformed_users} with $\cT_s := \cM_s^{\otimes l}$, an $(a, M_1,M_2)$ ET code 
			\begin{align}
			\tilde{\cC} \ = \ (\tilde{W}(m), \tilde{\cE}, \tilde{\cD}_m)_{m=1}^{M_1},
			\end{align}
			with
			\begin{align}
			\frac{1}{a} \log M_1 \     &\geq \ \underset{s \in S}{\inf} \ I(X;C^l, \omega_{s,l}) - \frac{\delta}{2}, \hspace{.5cm} \text{and} \hspace{.5cm}
			\frac{1}{a} \log M_2 \ 
			\geq \ \underset{s \in S}{\inf} \ I(B^l\rangle XC^l, \omega_{s,l}) - \frac{\delta}{2}.
			\end{align}
			The informational quantities on the right hand sides of the above inequalities are evaluated on the states
			\begin{align*}
			\omega_{s,l} \ := \ \sum_{x \in \cX} p(x) \ket{x} \bra{x} \otimes \id_{\cK_B}^{\otimes l} \otimes \cM_s^{\otimes l}(W(x) \otimes \Psi),  &&(s \in S)
			\end{align*}
			and moreover, 
			\begin{align}
			\underset{s \in S}{\inf} P^{ET}(\tilde{\cC}, \cM_s^{\otimes a\cdot l}) \ \geq \ 1 - 2^{-la\tilde{c}}
			\end{align}
			holds with a constant $\tilde{c} > 0$. We define another $(n,M_1,M_2)$ EG code $\cC := (W(m), \cE, \cD_m)_{m=1}^{M_1}$ with 
			\begin{align*}
			W(m) 
			&:= \tilde{W}(m) \otimes \pi_{\cK_A}^{\otimes b},  \\
			\cE(x) 
			&:= \tilde{\cE}(x) \otimes \pi_{\cK_B}^{\otimes b} &&(x \in \cL(\cF)), \ \text{and} \\
			\cD_m(y) \ 
			&:= \ \tilde{\cD}_m \otimes \tr_{\cH_C}^{\otimes b}(y) &&(y \in \cL(\cH_C^{\otimes n})).
			\end{align*}
			It holds for all $s \in S, m \in [M_1]$
			\begin{align}
			&F\left(\ket{m}\bra{m} \otimes \Phi, \id_{\cK_B}^{\otimes n} \otimes \cD \circ M_s^{\otimes n} \circ(\id_{\cK_A}^{\otimes n} \otimes \cE)(W(m) \otimes \Phi)\right) \\
			&\geq \ F\left(\ket{m}\bra{m} \otimes \Phi, \id_{\cK_B}^{\otimes n} \otimes \cD \circ M_s^{\otimes n} \circ(\id_{\cK_A}^{\otimes n} \otimes \cE)(W(m) \otimes \Phi)\right),
			\end{align}
			i.e. 
			\begin{align}
			P^{ET}(\cC, \cM_s^{\otimes n}, m ) \ \geq \ P^{ET}(\tilde{C}, \cM_s^{\otimes l\cdot a}) \geq 1 - 2^{-la\tilde{c}} = 1 - 2^{-n c}
			\end{align}
			with $c := \tilde{c}/(l+1)$. Moreover, if $n$ is large enough, we have 
			\begin{align}
			\frac{1}{n} \log M_1 \ 
			\geq \ \left(\frac{1}{a + l^{-1}}\right) \frac{1}{l} \log M_1  \geq \ \frac{1}{l} \left(\frac{1}{a}\log M_1 -\frac{\delta}{2}\right) 
			\geq \ \frac{1}{l} \underset{s \in S}{\inf} I(X;C^l, \omega_{s,l}) - \delta. \label{prop:inner_bound_no_csi_rate_1}
			\end{align}
			In the same manner, we can also show
			\begin{align}
			\frac{1}{n} \log M_2 \ \geq \ \frac{1}{l} \underset{s \in S}{\inf} \ I(B^l\rangle C^lX, \omega_{s,l}) - \delta. \label{prop:inner_bound_no_csi_rate_2}
			\end{align}
			With the inequalities in $(\ref{prop:inner_bound_no_csi_rate_1})$, and $(\ref{prop:inner_bound_no_csi_rate_2})$ we have shown
			\begin{align}
			\bigcap_{\cM \in \fM} \frac{1}{l}  \hat{C}^{(1)}(\cM^{\otimes l}, p, V, \Psi) \subset CQ^{ET}(\fM)_\delta.
			\end{align}
			Since $\delta > 0$ was arbitary, (\ref{prop:inner_bound_no_csi_2}) follows. 
		\end{proof}
	\end{subsection}
	\begin{subsection}{Outer bounds to the capacity regions} \label{subsect:proofs_outer_bounds}
		In this section, we prove the outer bounds to the capacity regions as stated in Theorem \ref{theorem:main}. 
		\begin{proposition}[Outer bound to the informed receiver capacity region]\label{prop:outer_bound_c_csi}
			Let $\fM \subset \cC(\cH_A \otimes \cH_B, \cH_C)$. It holds
			\begin{align}
			CQ^{I}(\fM) \ \subset \ \cl \left(\bigcup_{l=1}^\infty \bigcup_{p,V,\Psi} \bigcap_{\cM \in \fM} \frac{1}{l}  \hat{C}^{(1)}(\cM^{\otimes l}, p, V, \Psi) \right)
			\end{align}
		\end{proposition}
		\begin{proof}
			The proof of the statemt is fairly standard, and we give the argument for the reader's convenience. We will show that the proposed outer bound remains valid even in case, that the receiver is allowed to 
			choose the decoding channel dependent on the channel parameter, i.e. user $C$ has channel state information (CSI). 
			Fix an arbitrary $\delta > 0$. Let $\{\cC_{n,s}\}_{s \in S}$ with $\cC_{n,s} := (W(m), \Psi, \cD_{m,s})_{m=1}^{M_{1,n}}$ be an $(n,M_{1,n},M_{2,n})$ scenario-II code with $C$-CSI. We set
			\begin{align}
			\underset{s \in S}{\inf} \ P^{EG}(\cC_{n,s}, \cM_s^{\otimes n}) \ = \ 1 - \epsilon_n, \label{prop:outer_bound_c_csi_1}
			\end{align}
			and 
			$R_{n,i} := \frac{1}{n} \log M_{i,n}$ 
			For $i \in \{1,2\}$. Define moreover the states 
			\begin{align*}
			\omega_{MBC^n,s} = \omega(\cM^{\otimes n}_s, P_M, V, \Psi) = \sum_{i=1}^{M_{1,n}} P_M(m) \ket{m}\bra{m} \otimes (\id_{\cK_B} \otimes \cM^{\otimes n}_s)(\Phi(m) \otimes \Psi),
			\end{align*}
			and 
			$\omega'_{M\hat{M}B\hat{B},s} := (\id_{\bbmC^{M_{n,1}}} \otimes \id_{\cK_B} \otimes \cD_s)(\omega_{MBC^n,s})$
			for each $s \in S$, where $P_M$ is the equidistribution on the message set $[M_{1,n}]$, $\phi(m)$ is a purification of $V(m)$,  and $\cD_s(x) := \sum_{m=1}^M \otimes \ket{m}\bra{m} \otimes \cD_{m,s}(x)$ for each $s \in S$.
			We define a pair $(M,\hat{M}_s)$ of classical random variables having joint distribution 
			\begin{align*}
			P_{M\hat{M}_s}(m,m') := \braket{m \otimes m', \omega_{M\hat{M},s} m \otimes m'}. 
			\end{align*}
			It holds 
			\begin{align}
			\Pr\left(M \neq \hat{M}_s \right) \ \leq \ 2 \sqrt{\epsilon_n} := \tilde{\epsilon_n} \label{prop:outer_bound_c_csi_2}
			\end{align}
			by (\ref{prop:outer_bound_c_csi_1}). We have for each $s \in S$
			\begin{align}
			\log M_{1,n} \ 
			&= \ H(M) \nonumber \\
			&= \ I(M;\hat{M}_s) + H(M|\hat{M}_s) \nonumber \\
			&\leq I(M;\hat{M}_s) + \tilde{\epsilon_n} \log M_{1,n} + 1 \nonumber \\
			&\leq I(M;C^n, \omega_{MBC^n,s}) + \tilde{\epsilon_n} \log M_{1,n} + 1 \label{prop:outer_bound_c_csi_4}
			\end{align}
			where first of the above inequalities is by Fano's lemma, while the second one is Holevo's bound. With $\Phi$ being the target maximally entangled state of the quantum part of the code, we have, by monotonicity of the fidelity under 
			taking partial traces
			\begin{align*}
			P^{ET}(\cC_{n,s}, \cM^{\otimes n}_s) \ = \ \frac{1}{M_{1,n}} \sum_{m=1}^{M_{1,n}} \ F( \ket{m} \bra{m} \otimes \Phi, \omega'_{\hat{M}\hat{B}B_s}) \leq F(\Phi, \omega'_{B\hat{B},s}),
			\end{align*}
			which, combined with $(\ref{prop:outer_bound_c_csi_1})$ allows to bound
			$\|\Phi - \omega'_{\hat{B}B,s}\|_1 \leq \tilde{\epsilon}_n$.
			Using the Lemma \ref{lemma:alicki_fannes} in Appendix \ref{sect:auxiliary_results}, we have
			\begin{align}
			|I(B\rangle \hat{B}, \Phi) - I(B\rangle \hat{B}, \omega'_{\hat{B}B,s})| \ \leq \ 2 H(\tilde{\epsilon_n}) + 4 \tilde{\epsilon}_n \cdot \log M_{2,n} \label{prop:outer_bound_c_csi_3}. 
			\end{align}
			Consequently
			\begin{align}
			\log M_{2,n} \ 
			&= \ I(B \rangle \hat{B}, \Phi) \nonumber \\
			&\leq \ I(B \rangle \hat{B}, \omega'_{B\hat{B},s}) + 2 H(\tilde{\epsilon_n}) + 4 \tilde{\epsilon}_n \cdot \log M_{2,n} \nonumber \\
			&\leq \ I(B \rangle \hat{B}M, \omega'_{B\hat{B}M,s}) + 2 H(\tilde{\epsilon_n}) + 4 \tilde{\epsilon}_n \cdot \log M_{2,n} \nonumber \\
			&\leq \ I(B \rangle \hat{B}M, \omega'_{BC^nM,s}) + 2 H(\tilde{\epsilon_n}) + 4 \tilde{\epsilon}_n \cdot \log M_{2,n} \label{prop:outer_bound_c_csi_5}
			\end{align}
			If we now demand $\epsilon_n \rightarrow 0$, Eqns. (\ref{prop:outer_bound_c_csi_4}) and (\ref{prop:outer_bound_c_csi_5}) show, that if $n$ is large enough, the inequalities 
			\begin{align*}
			\frac{1}{n} \log M_{n,1} &\leq \inf_{s \in S} \frac{1}{n} \ I(M;C^n, \omega_{MBC^n,s}) + \delta, \ \text{and} \\
			\frac{1}{n} \log M_{n,2} &\leq \inf_{s \in S} \frac{1}{n} \ I(B \rangle \hat{B}M, \omega'_{BC^nM,s}) + \delta
			\end{align*}
			hold, which implies the chain of inclusion relations
			\begin{align*}
			(R_{1,n}, R_{2,n}) 
			& \in \ \bigcap_{s \in S} \frac{1}{n} \hat{C}^{(1)}(\cM_s^{\otimes n}, P_M, W, \Psi)_{\delta} \\
			& \subset \bigcup_{p, V, \Psi} \bigcap_{s \in S} \frac{1}{n} \hat{C}^{(1)}(\cM_s^{\otimes n}, p ,V ,\Psi)_\delta \\
			& \subset \bigcup_{l=1}^\infty\bigcup_{p, V, \Psi} \bigcap_{s \in S} \frac{1}{n} \hat{C}^{(1)}(\cM_s^{\otimes n}, p ,V ,\Psi)_\delta  \\
			& \subset \cl\left(\bigcup_{l=1}^\infty\bigcup_{p, V, \Psi} \bigcap_{s \in S} \frac{1}{n} \hat{C}^{(1)}(\cM_s^{\otimes n}, p ,V ,\Psi)\right)_\delta.
			\end{align*}
			Since $\delta > 0$ was arbitrary, we are done. 
		\end{proof}
	\end{subsection}
\end{section}
\begin{section}{Conclusion and Discussion}
	In this work, we derived a multi-letter description to the capacity region of the compound QMAC where one sender sends classical messages while the other aims to transmit quantum information. To our knowledge the characterization 
	in Theorem \ref{theorem:main} is the first result for coding of the QMAC where channel uncertainty is present and genuine quantum transmission tasks such as entanglement transmission and generation are performed. 
	(the earlier work \cite{hirche16} considered the case where all senders send classical messages.) Further research in this direction hopefully will bring new insights for coding of the compound QMAC also in situations where more senders 
	are present and full-quantum coding is performed.\newline 
	The argument employed to prove the coding theorem follows a strategy which seems rather common for perfectly known classical as well as quantum multiple-access channels. A clever combination of single-terminal random codes 
	results in a random code for the QMAC which has sufficient performance on the average. In quantum information theory, this technique was used for simultaneous classical and quantum coding over a single sender quantum channel and can 
	be generalized to channels with uncertainty also in this case \cite{boche17}. By extracting powerful universal random coding results for classical message and entanglement transmission from the literature (\cite{bjelakovic09}, \cite{mosonyi15}), we
	were able to make successful use of this strategy also in case of channel uncertainty. The mentioned universal random coding results may be applied in further research deriving codes in multi-user quantum information theory. 
	Beside consideration of more general multi-user and coding scenario also an extension of the results to other channel models may be a direction of further research. One possible variation of the compound MAC channel model is when one or 
	more of the users are provided with channel state information (CSI). Such additional knowledge of the system is already known as very relevant from the practical point of view when regarding classical channels. 
	Especially, as CSI might lead to substantially larger capacity regions. On the other hand, our results possibly provide the basis for tackling coding for the far more demanding arbitrarily varying QMAC (AVQMAC). In this model, the channel 
	statistics can be given by an arbitrary element chosen from a prescribed set for each channel use. An interpretation of this model is that the channel map is chosen by a malicious ``jamming'' party. The capacity region of the corresponding
	classical channel model was determined in \cite{jahn81}, \cite{ahlswede99}. The code construction established to show achievability in this work allows to achieve each point in the capacity regions with codes with fidelity 
	approaching one exponentially fast. This may allow to use Ahlswede's robustification and elimination techniques \cite{ahlswede78} to find good codes in case of the much more AVQMAC. \newline 
	We remark here, that albeit the characterization of the capacity regions $CQ^{EG}$ and $CQ^{EG}$ given in Theorem \ref{theorem:main} are correct, the description may be improved regarding computational aspects. While our capacity formula is a union
	over intersections of one-shot regions which are rectangles of points fulfilling 
	\begin{align}
	R_1 \ \leq I(X ; C^l, \omega_s)  \hspace{.5cm} \text{and} \hspace{.5cm} R_2 \ \leq \ I(B^l \rangle C^lX, \omega_s) 
	\end{align}
	where $\omega_s$ may be the state defined in (\ref{def:effective_cqq_state}) for each channel state $s \in S$. The authors of this paper feel, that a description replacing the rectangular one-shot regions with the above given boundaries 
	by pentagons with boundaries 
	\begin{align}
	R_1 \leq I(X\rangle C^l B^l, \omega_s), \hspace{.5cm} R_2 \leq I_c(B^l\rangle C^l X, \omega_s) \hspace{.5cm} \text{and} \hspace{.5cm} R_1 + R_2 \leq I(X;C^l, \omega_s) + I_c(B^l \rangle C^l X, \omega_s)
	\end{align}
	instead. Such a characterization was indicated to be possible in case of the perfectly known QMAC in \cite{yard08}, Chapter VII. Whether or not such a description is indeed possible for each example of a compound QMAC and to what extend it 
	would improve the capacity formula given here is a future research topic. In case of a perfectly known QMAC it was discussed in \cite{yard05} that standard examples having single-letter capacity regions in the latter characterization also 
	single-letterize in the first characterization.
\end{section}
\begin{appendix}
	\begin{section}{Auxiliary results}\label{sect:auxiliary_results}
		For the convenience of the reader, some auxiliary standard results used in the text are collected. 
		\begin{lemma}[Gentle measurement lemma \cite{winter99}] \label{lemma:gentle_measurement}
			Let $\rho \in \cS(\cK)$, $E \in \cL(\cH)$, $0 \leq E \leq \bbmeins_{\cH}$. It holds
			\begin{align*}
			\|\sqrt{E}\rho\sqrt{E} - \rho\|_1 \ \leq \ 3 \cdot \sqrt{1 - \tr E \rho}.
			\end{align*}
		\end{lemma}
		\begin{lemma}[\cite{yard08}] \label{lemma:auxiliary_results_yard_fidelity_1}
			Let $\Psi, \rho, \sigma \in \cS(\cK)$ be states, where $\Psi$ is pure. Then
			\begin{align*}
			F(\Psi, \rho) \ \geq \ F(\Psi, \sigma) - \tfrac{1}{2}\|\rho - \sigma\|_1
			\end{align*}
		\end{lemma}
		\begin{lemma}[\cite{yard08}] \label{lemma:auxiliary_results_yard_fidelity_2}
			Let $\Psi \in \cS(\cK_A), \rho \in \cS(\cK_B), \sigma \in \cS(\cK_{A} \otimes \cK_B)$ Then 
			\begin{align*}
			F(\Psi \otimes \rho, \sigma) \ \geq \ 1 - \|\rho - \sigma_B\|_1 - 3\left(1 - F(\Psi, \sigma_A)\right).
			\end{align*}
		\end{lemma}
		\begin{lemma}[\cite{alicki04},\cite{winter15}] \label{lemma:alicki_fannes}
			Let $\rho, \sigma \in \cS(\cK_{A} \otimes \cK_B)$, $\|\rho - \sigma\|_1 \leq \epsilon \leq 1$. Then
			\begin{align}
			|I(A\rangle B, \rho) - I(A \rangle B, \sigma)| \ \leq \ 6 \epsilon \log \dim \cK_A + (2 + 4 \epsilon) h(2 \epsilon  / 1 + 2 \epsilon),
			\end{align}
			where $h(x) = - x \log x - (1-x) \log (1-x)$, $s \in (0,1)$ is the binary Shannon entropy.
		\end{lemma}
	\end{section}
	\begin{section}{Convexity of the capacity formula} \label{appendix:capacity_regions_convexity}
		In this appendix, we show, that the functional expressions in Theorem \ref{theorem:main} do not need further convexification. Following the arguments given in \cite{yard08} for the case of a perfectly known QMAC, we show, that for 
		given $\fM \subset \cC(\cH_A \otimes \cH_B, \cH_C)$, the set
		\begin{align}
		\hat{C}_1(\fM) := \bigcup_{l=1}^\infty \bigcup_{p,V,\Psi} \bigcap_{\cM \in \fM} \hat{C}^{(1)}(\fM^{\otimes l}, p, V, \Psi), 
		\end{align}
		is convex, i.e. we show
		\begin{lemma}
			Let $\fM \in \cC(\cH_A \otimes \cH_B, \cH_C)$. It holds
			\begin{align}
			\conv(\hat{C}_1(\fM)) = \hat{C}_1(\fM).
			\end{align}
		\end{lemma}
		\begin{proof}
			We have to show, that for each $\lambda \in (0,1)$, and any two rate pairs $(R_1^{(i)},R_2^{(i)}) \in \hat{C}_1(\fM)$, $i = 1,2$, the their convex combination $(\overline{R}_1, \overline{R}_2)$ with $\overline{R}_j^{(i)} = 
			\lambda R_j^{(1)} + (1- \lambda) R_j^{(2)}$
			for $i = 1,2$ is also a member of $\hat{C}_1(\fM)$. Assume
			\begin{align}
			(R_1^{(i)}, R_2^{(i)}) \in \bigcap_{\cM \in \fM} \hat{C}^{(1)}(\cM^{\otimes l_i}, p_i, V_i, \Psi_i)
			\end{align}
			for some $l_i, p_i, V_i, \Psi_i$, i.e. with effective states 
			\begin{align}
			\omega_i(\cM) := \omega(\cM^{\otimes l_i}, p_i, V_i, \Psi_i) &&(i \in {1,2}, \cM \in \fM)
			\end{align}
			according to (\ref{def:effective_cqq_state}), the equations
			\begin{align}
			l_i \cdot R_1^{(i)} \ & \leq \ \underset{\cM \in \fM}{\inf} \ I(X_i; C^{l_i}, \omega_i), \ \text{and}   \nonumber \\
			l_i \cdot R_2^{(i)} \ & \leq \ \underset{\cM \in \fM}{\inf} \ I(B^{l_i}\rangle C^{l_iX_i}, \omega_i)   \label{appendix:capacity_regions_convexity_1}
			\end{align}
			are fulfilled. Fix $\delta > 0$, and let $k,n \in \bbmN$, $0 < k < n$ such that
			\begin{align}
			|\lambda - \tfrac{k}{n}| \ < \frac{\delta}{R_1^{(1)} + R_2^{(1)} + R_1^{(2)} + R_2^{(2)}}.
			\end{align}
			Set $t_1 := k l_2$, $t_2 := (n-k) l_1$. 
			With 
			\begin{align}
			\tilde{\omega}(\cM) \ := \ \omega_1(\cM)^{\otimes t_1} \otimes \omega_2(\cM)^{\otimes t_2}, 
			\end{align}
			which is unitarily equivalent to 
			\begin{align}
			\omega(\cM^{\otimes nl_1l_2}, p_1^{t_2}\otimes p_2^{t_2}, V_1^{\otimes t_1} \otimes V_2^{\otimes t_2}, \Psi_1^{\otimes t_1} \otimes \Psi_2^{\otimes t_2})
			\end{align}
			we have
			\begin{align}
			I(X_1^{t_1}X_2^{t_2}; C^{l_1l_2n}, \tilde{\omega}(\cM))
			& \ = \ I(X_1^{t_1}X_2^{t_2}; C^{l_1l_2n}, \omega_1(\cM)^{\otimes t_1} \otimes \omega_2(\cM)^{\otimes t_2}) \\
			& \ = \ I(X_1^{t_1}; C^{l_1t_1}, \omega_1(\cM)^{\otimes t_1}) + I(X_2^{t_2}; C^{l_2t_2}, \omega_2(\cM)^{\otimes t_2}) \\
			& \ = \ t_1 \cdot I(X_1; C^{l_1}, \omega_1(\cM)) + t_2 \cdot I(X_2; C^{l_2}, \omega_2(\cM)) \\
			& \ \geq t_1 l_1 R_1^{(1)} + t_2 l_2 R_1^{(2)} \\
			& \ \geq k l_1 l_2 R_1^{(1)} + (n-k) l_1 l_2 R_1^{(2)}.
			\end{align}
			where the second and third equality above are by additivity of the quantum mutual information evaluated on product states, and the inequality is by (\ref{appendix:capacity_regions_convexity_1}). Consequently, we have
			\begin{align}
			\frac{1}{l_1 l_2 n} \underset{\cM \in \fM}{\inf} I(X_1^{t_1}X_2^{t_2}; C^{l_1l_2n}, \tilde{\omega}(\cM)) 
			& \ \geq \  \frac{k}{n} R_1^{(1)} + (1 - \frac{k}{n} R_1^{(2)}  \\
			& \ \geq \ \lambda R_1^{(1)} + (1-\lambda) R_1^{(2)} - \delta.  \label{appendix:capacity_regions_convexity_2}
			\end{align}
			In a similar manner, also the inequality 
			\begin{align}
			\frac{1}{l_1 l_2 n} \underset{\cM \in \fM}{\inf} I(B^{n l_1 l_2} \rangle C^{n l_1 l_2}X_1^{t_1}X_2^{t_2}, \tilde{\omega}(\cM)) & \ \geq \ \lambda R_2^{(1)} + (1-\lambda) R_2^{(2)} - \delta \label{appendix:capacity_regions_convexity_3}
			\end{align}
			is verified. By (\ref{appendix:capacity_regions_convexity_2}), and (\ref{appendix:capacity_regions_convexity_3}), 
			\begin{align}
			(\overline{R}_1, \overline{R}_2) \ 
			\in \ \left( \bigcap_{\cM \in \fM} \frac{1}{l_1l_2} \hat{C}^{(1)}(\cM^{l_1l_2n}, p_1^{t_1}\otimes p_2^{t_2}, V_1^{\otimes t_1} \otimes V_2^{\otimes t_2}, \Psi_1^{\otimes t_1} \otimes \Psi_2^{\otimes t_2})\right)_\delta \
			\subset \ \hat{C}_1(\fM)_\delta.
			\end{align}
			Since $\delta > 0$ can be chosen arbitrarily, the claim of the lemma follows. 
		\end{proof}
	\end{section}
\end{appendix}




\end{document}